\documentclass[journal]{IEEEtran}
\usepackage{verbatim}
\usepackage{amsfonts}
\usepackage{amssymb}
\usepackage{stfloats}
\usepackage{cite}
\usepackage{graphicx}
\usepackage{psfrag}
\usepackage{subfigure}
\usepackage{amsmath}
\usepackage{array}
\usepackage{epstopdf}
\usepackage{authblk}
\usepackage{graphicx} 
\usepackage{amsthm} 
\usepackage{lipsum}
\usepackage{verbatim} 
\usepackage{authblk}
\usepackage{mathtools}
\usepackage{cuted}
\usepackage[lined,boxed,ruled]{algorithm2e}
\usepackage{booktabs}
\usepackage{subfigure}
\usepackage{siunitx}
\usepackage{mathtools}
\usepackage{soul} 
\usepackage[]{mdframed}
\usepackage{setspace}
\usepackage{algpseudocode}
\usepackage[colorlinks,linkcolor=black,anchorcolor=black,citecolor=black,urlcolor=black]{hyperref}
\usepackage{enumitem}

\newtheorem{Theorem}{Theorem}

\newtheorem{Proposition}{Proposition}
\newtheorem{Remark}{Remark}

\newtheorem{Corollary}{Corollary}

\title{Optimal Expert Selection for Distributed Mixture-of-Experts at the Wireless Edge}
\author{}

\makeatletter
\newcommand{\removelatexerror}{\let\@latex@error\@gobble}
\makeatother

\begin{document}
\author{{Shengling~Qin},~{Hai~Wu},~{Hongyang~Du}, and~{Kaibin~Huang}

\thanks{S. Qin, H. Wu, H. Du and K. Huang are with the Department of Electrical and Electronic Engineering, The University of Hong Kong, Hong Kong (Email: \{slqin, wuhai, duhy, huangkb\}@eee.hku.hk). Corresponding author: H. Du.}}
\maketitle

\begin{abstract}

The emergence of distributed Mixture-of-Experts (DMoE) systems, which deploy expert models at edge nodes, offers a pathway to achieving connected intelligence in sixth-generation (6G) mobile networks and edge artificial intelligence (AI).
However, current DMoE systems lack an effective expert selection algorithm to address the simultaneous task-expert relevance and channel diversity inherent in these systems.
Traditional AI or communication systems focus on either performance or channel conditions, and direct application of these methods leads to high communication overhead or low performance.
To address this, we propose the DMoE protocol to schedule the expert inference and inter-expert transmission. This protocol identifies expert selection and subcarrier allocation as key optimization problems.
We formulate an expert selection problem by incorporating both AI performance and channel conditions, and further extend it to a Joint Expert and Subcarrier Allocation (JESA) problem for comprehensive AI and channel management within the DMoE framework.
For the NP-hard expert selection problem, we introduce the Dynamic Expert Selection (DES) algorithm, which leverages a linear relaxation as a bounding criterion to significantly reduce search complexity.
For the JESA problem, we discover a unique structural property that ensures asymptotic optimality in most scenarios. We propose an iterative algorithm that addresses subcarrier allocation as a subproblem and integrates it with the DES algorithm.
The proposed framework effectively manages the tradeoff between task relevance and channel conditions through a tunable importance factor, enabling flexible adaptation to diverse scenarios.
Numerical experiments validate the dual benefits of the proposed expert selection algorithm: high performance and significantly reduced cost. JESA consistently achieves higher accuracy compared to homogeneous expert selection and lowers the cost by up to 50\% compared to Top-k scheduling.

\end{abstract}

\begin{IEEEkeywords}
Mixture-of-Experts, edge AI, radio resource management, multi-domain tasks
\end{IEEEkeywords}

\section{Introduction}
The Mixture-of-Experts (MoE) is a new architecture for large language models advocated by numerous high-tech companies, e.g., OpenAI, DeepSeek, and Microsoft~\cite{jacobs1991adaptivemoe,jordan1994hierarchicalmoe,shazeer2017outrageouslymoe,liu2024deepseek}. An MoE model comprises a set of sub-models, each of which is trained for a specific task and called an \emph{expert}. The MoE models are gaining increasing popularity due to their effectives in executing multi-domain tasks while reining in computation complexity via activating only a subset of task-relevant experts. Meanwhile, edge AI, referring to machine learning and AI algorithms deployed distributively at the sixth-generation (6G) mobile network edge, offers a promising platform to deploy experts at the network edge to deliver intelligent services to users~\cite{edge2023letaief,inference2023shao,edge2020li,communication2020shao,liu2023earlyexit}. MoE is well-suited for edge AI as the experts in an MoE model can be easily separated and distributed at different resource-constrained nodes (edge devices or servers), creating the new paradigm of distributed MoE (DMoE)~\cite{xue2024wdmoe}. Then given a specific task, a wireless-connected expert can be activated depending on its task-relevance level as well as channel state. Enabling an efficient and effective DMoE system calls for developing a new class of \emph{radio resource management} (RRM) techniques with awareness of experts' task-relevance, which is a largely uncharted area and the theme of this work.

A DMoE system is characterized by the co-existence of two types of diversity -- \emph{expertise and channel diversity}. First, consider expertise diversity. A typical AI model exhibits a varying level of specialization across different domains. For instance, a math model may achieve high accuracy in solving math problems, whereas a medical model could yield poor results for the same task. As a result, the model shows significantly varying performance for different tasks, corresponding to different levels of task-model relevance. Then, given a task, the variation of relevance across distributed experts is called \emph{expertise diversity}~\cite{wu2024mba}. On the other hand, the classic notion of channel diversity refers to the variation of channel state across wireless links~\cite{mecsurvey}. To maintain a high Quality-of-Service (QoS) standard, the execution of a task in a DMoE system requires selecting not only the most relevant experts but also those with favorable channel conditions. This requires designing new RRM techniques that exploit both expert and channel diversity. The design approach is aligned with the main trend of integrated AI and communications in the area of edge AI~\cite{liu2023earlyexit,wu2024mba}.

RRM for DMoE systems is related to that for Mobile Edge Computing (MEC) systems as both involve cross-disciplinary designs integrating computing and communication. Nevertheless, RRM designs for MEC typically assumed generic computing tasks based on abstracted computation models (e.g., a task is characterized by its required number of processor clock cycles)~\cite{mecsurvey}. Such assumptions and models allow tractable designs for allocating radio resources (e.g., bandwidth, subcarriers, timeslots, and power) to minimize latency or energy consumption~\cite{tradeoff2018li}. For instance, an energy-efficient MEC framework is derived in~\cite{mecenergy} for optimal joint allocation of subcarrier-and-computational resources. Another approach as proposed in~\cite{saleem2020mecd2d} allocates spectrum and CPU resources, enabling devices to either partially offload computational tasks to the edge node or leverage nearby devices’ computational resources. In general, the existing MEC studies largely assume uniform performance across edge nodes for a specific task if their resources are identical; task-model relevance is mostly overlooked in the literature~\cite{mecsurvey}. 

Recent works in the area of edge AI have begun to address the diverse capabilities of AI models~\cite{liu2023earlyexit,xue2024wdmoe}. A recent study of DMoE presents a heuristic scheme to select expert models for a task based on the performance-to-latency ratio~\cite{xue2024wdmoe}. The consideration of task-model relevance is also reflected in the early-exiting inference scheme that halts the inference process based on the current model’s performance~\cite{liu2023earlyexit}. These early works rely on simplified or heuristic performance metrics and lack a systematic design approach with fine-grained joint optimization considering both expertise and channel diversity. Moreover, though the area of edge AI has a rich literature, many studies still rely on either abstract computational models or traditional deep neural networks (DNNs), resulting in the sub-optimality of relevant RRM strategies for the latest MoE architecture (see, e.g., \cite{wu2024mba,mecenergy}).

Classic RRM techniques have been extensively researched in the context of traditional wireless networks. These methods primarily focus on optimizing the allocation of power and bandwidth resources to enhance traditional network performance metrics, such as throughput, latency, and energy efficiency~\cite{you2016mec,you2017mec}.
The problem of optimal subcarrier and power allocation, along with rate control, has been solved in various contexts, including wireless communications in orthogonal frequency-division multiple access (OFDMA) systems~\cite{kivanc2003allocation}, small cell networks~\cite{zhang2018allocation}, broadband networks~\cite{niyato2006allocation}, and for proportional fair scheduling~\cite{Radunovic2007fairness, dai2009fairness}.
A range of solution methods have been developed, including greedy algorithm~\cite{kivanc2003allocation}, convex programming~\cite{zhang2018allocation}, Markov Decision Problem (MDP)~\cite{sun2021allocation}, machine learning and reinforcement learning~\cite{ceran2019allocation}, to solve the RRM problems.
While these classic RRM techniques provide a robust foundation for efficient resource utilization in traditional wireless networks, they are not directly transferable to DMoE systems.
The key distinction lies in the underlying paradigms. In traditional networks, scheduling schemes primarily affect time and energy consumption without impacting the correctness of executed tasks. In contrast, in DMoE systems, resource allocation directly influences not only the efficiency but also the correctness and quality of AI-driven tasks. This intricate interdependence between resource management and AI task performance presents unique challenges, underscoring the need for innovative RRM strategies that seamlessly integrate AI considerations with traditional metrics to achieve optimal system performance.


Specifically, we consider the RRM problem in a DMoE system, where multiple experts are deployed at edge nodes to collaboratively process user tasks. During inference, the tasks are not handled by a single expert but by a selected subset of experts to leverage their diverse knowledge. This necessitates data-intensive inter-expert communication and overcoming two resultant RRM challenges is the focus of this work: (1) selecting the appropriate experts (the expert selection problem) and (2) allocating channel resources to them (the joint expert and subcarrier allocation problem). These challenges differ from classic RRM problems because they interweave expertise and channel diversity. On the other hand, in centralized MoE where the model is deployed at a single node, channel states are not a concern with the task-relevance being the only consideration, leading to simple selection strategies like Top-k. On the contrary, in traditional generic MEC, task-relevance and other AI metrics are not considered and channel conditions dominate the selection process.
Then the main challenge faced by RRM for DMoE systems is to balance the following fundamental tradeoff. Given a task, allocating more resources to experts with high task-relevance levels may select those with unfavorable channel states but allocation based on experts' channel states may end up with those with mismatched expertise. To balance the tradeoff, we present a joint optimization framework to exploit both expertise and channel diversity in the DMoE system.
The main contributions of this work are described as follows.

\begin{enumerate}
    \item \textbf{DMoE Protocol}: We first design the DMoE protocol to enable the system of distributed experts to serve user tasks. With an MoE model partitioned and assigned to edge nodes to form experts, the user tasks, represented by queries consisting of tokens, are sent to a server that coordinates the experts to collaboratively perform inference. The processing of each query requires $L$ rounds of expert inference and inter-expert transmission, where each round involves executing a layer of the MoE model.
    Within each round, the server selects the most suitable experts for hidden states to leverage their domain-specific knowledge and allocates subcarriers for inter-expert transmission.
    The results are then aggregated to produce an output that integrates multi-domain knowledge. The exchange of high-dimensional hidden states dominates the communication overhead, making the expert selection and subcarrier allocation problems the focus of system optimization;
    \item \textbf{Optimal Expert Selection}: First, under the assumption of equal bandwidth allocation, our goal is to select experts that meet certain task-relevance requirements and also have favorable channel states to minimize the sum energy consumption. The formulated expert selection problem incorporates layer-wise QoS requirements representing task-relevance. Channel conditions, on the other hand, are integrated into the transmission energy term to denote the selection cost. To solve the NP-hard problem, we propose the optimal Dynamic Expert Selection (DES) algorithm, which is based on tree search and features a linear-relaxed lower bound as the bounding criterion to reduce search complexity. The bound is derived by identifying a critical expert and excluding less significant experts based on the energy-to-task-relevance order;
    \item \textbf{Joint Expert and Subcarrier Allocation} (JESA): We further formulate a JESA problem to manage both experts and subcarriers to minimize the overall energy consumption overhead. To solve this NP-hard problem, a joint optimization strategy is developed by integrating the task-relevance analysis from the preceding step and the optimal subcarrier allocation method. The direct iteration between expert selection and subcarrier allocation steps does not guarantee optimality, however, through an in-depth examination of the problem structure, we discover a unique property that guarantees asymptotic optimality. Leveraging this property, the block coordinate descent approach is adopted to achieve the asymptotic optimality.
    \item \textbf{Experiments}: We validate the effectiveness of the DMoE framework through extensive experiments with real datasets. The results demonstrate that our approach achieves higher accuracies compared to depth-unaware expert selection and incurs lower energy consumption compared to Top-k expert selection.
\end{enumerate}

The rest of this paper is organized as follows. Section II introduces the system model. Section III addresses the DMoE protocol. Section IV introduces the expert selection and the JESA problems. Sections V and VI present the solutions for the two problems, respectively, yielding the DES and JESA algorithms. Experimental results are provided in Section VII and concluding remarks in Section VIII.

\section{System Model}

\begin{figure*}[t]
    \centering
    \includegraphics[width=1.0\linewidth]{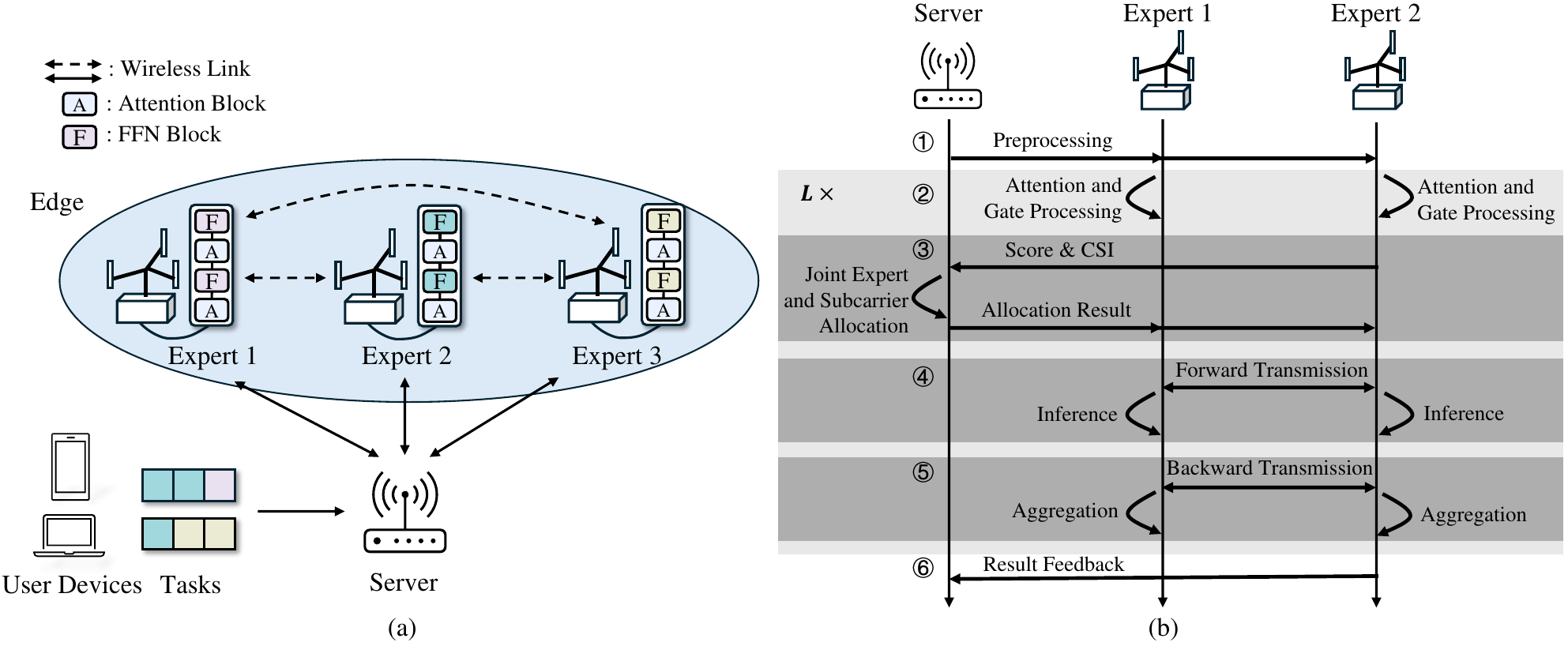}
    \caption{(a) The DMoE system; (b) The DMoE protocol.}
    \label{fig:system}
\end{figure*}

Consider the DMoE system depicted in Fig.~\ref{fig:system}(a). The system comprises $K$ expert nodes, an edge server, and multiple user devices. Users generate tasks with each task represented as a query of multiple tokens.
In this setup, users send queries to the server, which coordinates experts to collaboratively process the queries through expert inference and inter-expert transmission.
Experts are interconnected via device-to-device links.
Tokens processed by an expert become hidden states~\cite{fedus2022switch} and are directed to other suitable experts, leveraging their knowledge to assist inference. The results are then sent back to the originating expert for aggregation to produce an output that integrates multi-domain knowledge, in a Mixture-of-Experts manner. Multiple rounds are performed to traverse the entire global model.
Relevant models and metrics are introduced as follows.

\subsection{Communication Model}

In the DMoE system, two types of transmission are included: hidden state transmission between experts, and query uploading and result downloading between the server and the user.
Compared to the lightweight textual transmission in the server-user link, the high dimensionality of inter-expert hidden state transmission is the dominant communication overhead in the DMoE system, thus we only consider inter-expert transmissions.
We adopt OFDMA to facilitate multiaccess for selected experts in a single round. Considering a total of $K$ experts in the system, let $i = 1, 2, \ldots, K$ denote the sending expert, and $j = 1, 2, \ldots, K$ denote the receiving expert.
Note that when $i = j$, in-situ inference is performed, meaning the hidden states are processed locally without transmission overhead, making it the most efficient case. Since local computation does not incur communication costs, our focus is on optimizing inter-expert transmissions for $i \ne j$, which represent the primary source of communication overhead.
Let $B_0$ represent the subcarrier spacing, and $P_0$ the transmission power per subcarrier. Assume there are $M$ subcarriers available for allocation. The channel gain between expert $i$ and $j$ on subcarrier $m$ is denoted as $H_{ij}^{(m)}$. Interference is eliminated due to the exclusive subcarrier allocation. The maximum achievable data rate from expert $i$ to $j$ through subcarrier $m$ can be expressed as
\begin{equation}
r_{ij}^{(m)} = B_0 \log_2 \left( 1+\frac{H_{ij}^{(m)} P_0}{N_0} \right)
\end{equation}
where $N_0$ is the white Gaussian noise power.
To represent the utilization of subcarriers, we introduce the subcarrier assignment indicator $\beta_{ij}^{(m)}$ for $i \neq j$, where $\beta_{ij}^{(m)} = 1$ indicates that subcarrier $m$ is allocated to the transmission from expert $i$ to $j$, and $\beta_{ij}^{(m)} = 0$ otherwise.
The total maximum achievable data rate from expert $i$ to $j$ is given by
\begin{equation}
    R_{ij} = \sum_{m=1}^{M} \beta_{ij}^{(m)} r_{ij}^{(m)}.
\end{equation}

\subsection{Energy Consumption Models}

Let $s_0$ denote the size of one embedded hidden state for a token. For example, $s_0 = 8$ kB for a hidden state with 4096 dimensions and FP16 data format. The data size scheduled for transmission from expert $i$ to $j$ is $s_{ij} = s_0 \sum_{n=1}^{N_i} \alpha_{ij}^{(n)}$.
Then the communication energy~\cite{mecenergy} required to transmit all scheduled hidden states from expert $i$ to $j$ ($i \neq j$) is given by
\begin{equation}
\label{eq:ecomm}
    E_{ij}^{comm} = \frac{s_{ij}}{R_{ij}} \sum_{m=1}^M \beta_{ij}^{(m)} P_0.
\end{equation}
Previous studies~\cite{xugpunew} have profiled the energy consumption of GPUs when batch processing machine learning inference tasks. These profiling results show that energy consumption increases linearly with batch size, so the computation energy consumption of expert $j$ can be expressed as
\begin{equation}
    E_{j}^{comp} = a_{j} \sum_{i=1}^K s_{ij} + b_{j}
\end{equation}
where $a_{j} > 0$ and $b_{j} \geq 0$ are constants specific to device $j$.

\section{Overview of Distributed Mixture-of-Experts}

In this section, we begin by initializing the DMoE system and subsequently evaluating the diversity of expertise within the system. Following this, we present the DMoE protocol, which facilitates expert coordination and processes user tasks.

\subsection{System Initialization}
\label{sec:initial}

\begin{figure}[t]
    \centering
    \includegraphics[width=0.5\linewidth]{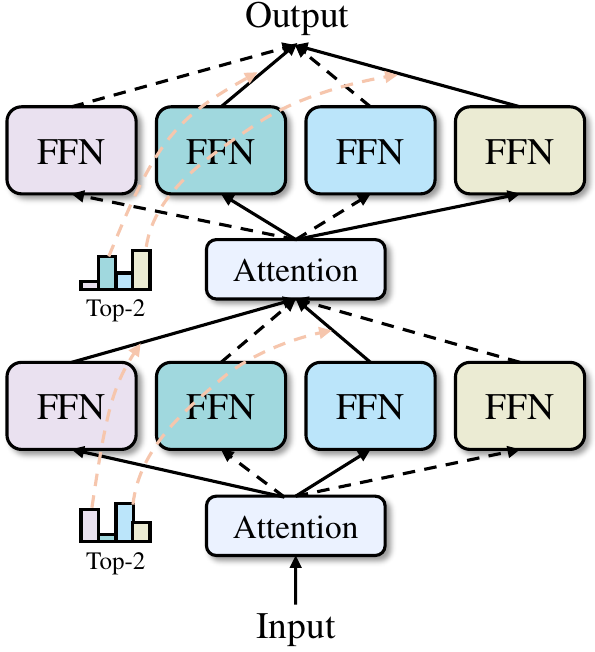}
    \caption{The architecture of a Mixture-of-Expert model.}
    \label{fig:llm_moe}
\end{figure}

The one-shot initialization is to distribute the MoE model to edge nodes for later inference. The initialization process consists of MoE vertical partitioning, block downloading, and expert model assembling.
\begin{enumerate}
    \item \textit{MoE Vertical Partitioning}:
    To initialize the DMoE system, a pre-trained MoE model is partitioned and distributed to the edge. Current state-of-the-art MoE models (e.g. Mixtral~\cite{jiang2024mixtral}, Deepseek MoE~\cite{liu2024deepseek}) all adopt a stacked Transformer architecture. The decoder-only MoE model mainly consists of $L$ stacked decoder layers, with layer index $l = 1, 2, \ldots, L$.
    As depicted in Fig.~\ref{fig:llm_moe}, each layer of the MoE model contains an attention block and several FFN blocks.
    With $\mathbf{Attn}^{(l)}$ denoting the attention block at layer $l$ and $\mathbf{FFN}^{(l)}_{i}$ the $i$-th FFN block at layer $l$, the MoE model can be expressed as
    \begin{equation}
        \begin{split}
            \mathbf{MoE} = & \left\{ \left[ \mathbf{Attn}^{(1)}, \left( \mathbf{FFN}^{(1)}_{1}, \mathbf{FFN}^{(1)}_{2}, \ldots \right) \right], \right. \\
            & \left. \left[ \mathbf{Attn}^{(2)}, \left( \mathbf{FFN}^{(2)}_{1}, \mathbf{FFN}^{(2)}_{2}, \ldots \right) \right], \ldots, \right. \\
            & \left. \left[ \mathbf{Attn}^{(L)}, \left( \mathbf{FFN}^{(L)}_{1}, \mathbf{FFN}^{(L)}_{2}, \ldots \right) \right] \right\}.
        \end{split}
    \end{equation}
    For edge distribution, the MoE model is partitioned into separate blocks, including the attention blocks $\mathbf{Attn}^{(l)}$, and the FFN blocks $\mathbf{FFN}^{(l)}_{i}$. The FFN blocks from all layers with the same index, e.g., $\mathbf{FFN}^{(1)}_{i}, \mathbf{FFN}^{(2)}_{i}, \ldots$, are grouped together with the shared attention blocks $\mathbf{Attn}^{(l)}$ as an expert. In this way, the MoE model is vertically partitioned into individual experts.

    \item \textit{Block Downloading}: After partitioning the MoE model into individual blocks, these blocks are transmitted to edge nodes.
    For node $i$, a set of blocks necessary for expert construction is downloaded, including the attention blocks $\mathbf{Attn}^{(l)}$ and the FFN blocks $\mathbf{FFN}^{(l)}_{i}$ for all layers $l$. These blocks are then stored locally on the edge. The downloading process is conducted in a one-shot manner, and the blocks are not updated during the inference stage. The gate function is also downloaded to the nodes.
    \item \textit{Expert Assembling}: The expert model is assembled by combining the attention block and the FFN block from each layer, and linking the layers sequentially. After the assembling, the expert model is ready for inference. Since all experts are partitioned from the same MoE model and constructed in the same way, the expert models share the same parameter structure.
    The assembled expert model on node $i$ can be expressed as
    \begin{equation}
    \begin{split}
        \mathbf{Expert}_{i} = & \left\{ \left[ \mathbf{Attn}^{(1)}, \mathbf{FFN}^{(1)}_{i} \right], \right.\\
        & \left. \left[ \mathbf{Attn}^{(2)}, \mathbf{FFN}^{(2)}_{i} \right], \ldots, \right. \\
        & \left. \left[ \mathbf{Attn}^{(L)}, \mathbf{FFN}^{(L)}_{i} \right] \right\}.
    \end{split}
    \end{equation}
\end{enumerate}

\begin{Remark}[Distribution Methods]
\emph{There are several methods to distribute MoE to the edge. In \cite{xue2024wdmoe}, the attention and FFN blocks are separately stored on the server and edge nodes. 
This approach introduces frequent communication from the server to edge nodes, as the attention results at the server need to be transmitted to the edge.
In comparison, the proposed distribution method assigns a whole set of attention and FFN blocks to an edge node to form an expert. By eliminating server-edge hidden state transmissions, our approach significantly reduces communication overhead.}
\end{Remark}

\subsection{Expertise Diversity}

\begin{figure}[t]
    \centering
    \includegraphics[width=0.8\linewidth]{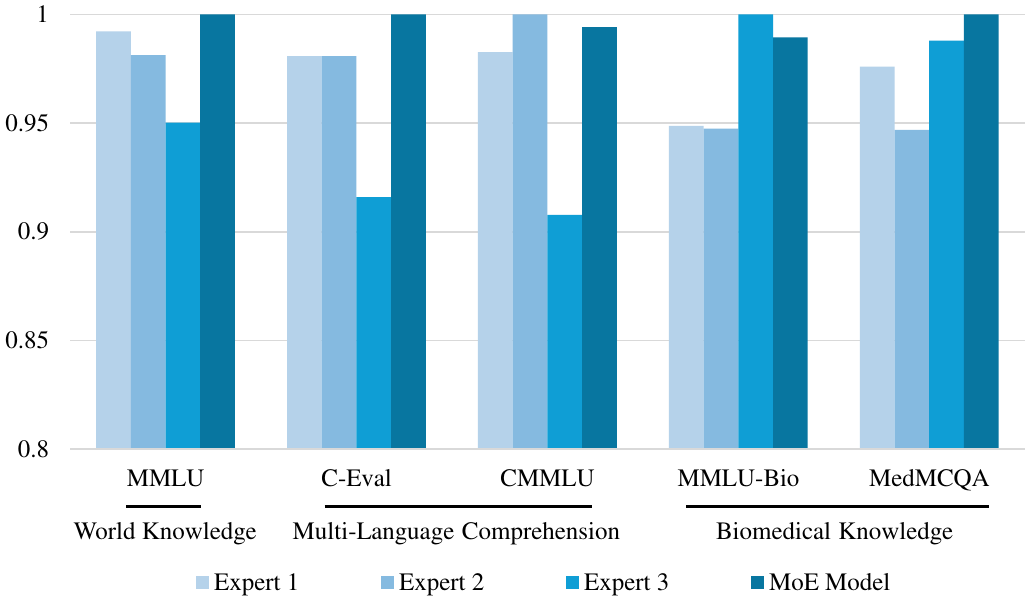}
    \caption{The normalized performance of the MoE model and individual experts across multi-domain tasks. The expertise diversity is observed.}
    \label{fig:diversity}
\end{figure}

We designed an experiment to assess whether the experts constructed from the original MoE model could inherit multi-domain specialization and exhibit expertise diversity on diverse tasks. For this experiment, we select Llama-3 series fine-tuned expert models, including Llama-3-8B-Instruct~\cite{llama3instruct} for general knowledge, Llama3-8B-Chinese-Chat~\cite{llama3chinese} for Chinese Q\&A, and Llama3-OpenBioLLM-8B~\cite{llama3bio} for medical-domain tasks. All models are available on Hugging Face. The gates are derived using the positive/negative prompt method.
The performance of the MoE model and individual experts is presented in Fig.~\ref{fig:diversity}. The result shows that the constructed experts successfully acquire specialized abilities from the MoE model, exhibiting different levels of specialization in general world knowledge (MMLU~\cite{hendrycks2021mmlu}), Chinese language comprehension (C-Eval~\cite{huang2024ceval} and CMMLU), and biomedical expertise (MMLU-Bio~\cite{hendrycks2021mmlu} and MedMCQA~\cite{pmlr-v174-pal22a-medmcqa}).
This experiment demonstrates that the partitioned experts have heterogeneous performance and need careful expert selection to match the task diversity.


\subsection{DMoE Protocol}

\begin{figure}[t]
    \centering
    \includegraphics[width=0.5\linewidth]{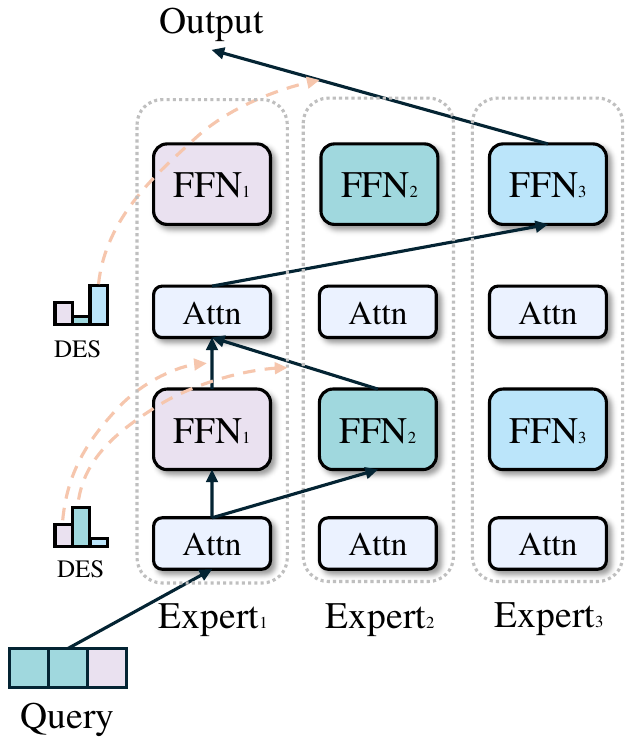}
    \caption{The expert selection process in DMoE.}
    \label{fig:compare}
\end{figure}

After initializing the experts, we propose the DMoE protocol to coordinate experts and process user tasks. As shown in Fig.~\ref{fig:system}(b), DMoE protocol mainly consists of attention and gate processing, joint expert and subcarrier allocation, forward transmission and inference, and backward transmission. This process repeats $L$ times for the hidden states to go through each layer. The detailed procedure of the DMoE protocol is listed as follows.

\begin{enumerate}
    \item \textit{Preprocessing}: The users first upload queries to the server. The server assigns and transmits the queries exclusively to the experts, with each expert assigned at most one query. Assume expert $i$ is assigned a query containing $N_i$ tokens. The tokens are then embedded into hidden states for further processing, with each corresponding to one token.
    \item \textit{Attention and Gate Processing}: The hidden states go through the attention block to extract features and check dependencies between tokens, allowing it to “attend” to the most relevant parts of the input~\cite{jacobs1991adaptivemoe,fedus2022switch,lepikhin2021gshard}.
    After attention processing, we denote the $n$-th hidden state of expert $i$ as $u_{i}^{(n)}$, with $\mathbf{\mathcal{U}}_{i}=\left \{u_{i}^{(1)},u_{i}^{(2)},...,u_{i}^{(N_i)} \right \}$ denoting all hidden states corresponding to the query.
    Then the hidden states are sent to the gate function.
    The gate function gives scores to estimate the performance of each expert with respect to each hidden state.
    With $g_{j}^{(l)}\left( u_{i}^{(n)} \right)$ regarded as the score of expert $j$ at layer $l$ processing $u_{i}^{(n)}$, the gate function of layer $l$ is defined as
    \begin{equation}
    \begin{split}
    \mathbf{g}^{(l)} \left(u_{i}^{(n)} \right) = & \left[ g_{1}^{(l)}\left( u_{i}^{(n)} \right),  g_{2}^{(l)}\left( u_{i}^{(n)} \right), \right. \\
    & \left. \ldots,  g_{j}^{(l)}\left( u_{i}^{(n)} \right),  \ldots, g_{K}^{(l)}\left( u_{i}^{(n)} \right) \right],
    \end{split}
    \end{equation}
    where $g_{j}^{(l)}\left( u_{i}^{(n)} \right) \geq 0$, $\forall i$ and $\sum_{j=1}^K g_{j}^{(l)}\left( u_{i}^{(n)} \right) = 1$.
    \item \textit{Joint Expert and Subcarrier Allocation}: The experts first upload the gating scores and inter-expert channel state information to the server. The server performs joint expert and subcarrier allocation to decide the expert assignment for hidden states $u_i^{(n)}$, and the channel usage to facilitate the inter-expert transmissions. For each $u_i^{(n)}$, a certain number of experts are selected to process it.
    We introduce the expert selection indicator $\alpha_{ij}^{(n)}$ to represent the expert allocation, where $\alpha_{ij}^{(n)} = 1$ indicates that expert $j$ is selected for $u_{i}^{(n)}$, and $\alpha_{ij}^{(n)} = 0$ indicates it is not.
    Then this step becomes deciding $\alpha_{ij}^{(n)}$ and $\beta_{ij}^{(m)}$.
    The Top-k is a commonly adopted expert selection strategy in centralized MoE focusing solely on the gating score. In DMoE, however, we aim to minimize the energy consumption while enforcing task-relevance requirements. We propose the Dynamic Expert Selection (DES) based on the task-relevance (i.e., the gating score) and the channel conditions. For subcarrier allocation, the server assign several subcarriers to a link between experts for transmission. The subcarrier allocation is to minimize the energy consumption of the inter-expert transmission. The expert and subcarrier allocation results are then sent back to the experts.
    \item \textit{Forward Transmission and Inference}: The hidden states are transmitted from the source expert to the selected expert according to the selection result. The selected experts leverage the FFN blocks to process hidden states from all requesting experts, resulting in updated hidden states of the same size.
    The domain-specific knowledge is mainly stored in the FFN blocks~\cite{jacobs1991adaptivemoe,fedus2022switch,lepikhin2021gshard}, thus combining the results from multiple experts can enhance the inference performance.
    Assuming the FFN block of expert $j$ at layer $l_i$ is represented by the function $\mathbf{FFN}_{j}^{(l_i)}(\cdot)$, the inference result of $u_{i}^{(n)}$ can be expressed as $\mathbf{FFN}_{j}^{(l_i)} \left( u_{i}^{(n)} \right)$.
    The updated hidden states have gained domain-specific knowledge from the experts.
    \item \textit{Backward Transmission and Aggregation}: The updated hidden states are transmitted back to the source expert. The results are then aggregated with the gating scores.
    The aggregation result of $u_{i}^{(n)}$ can be expressed as
    \begin{equation}
    \sum_{j=1}^K \frac {\alpha_{ij}^{(n)} g_{j}^{(l_i)}\left( u_{i}^{(n)} \right)}{\sum_{j=1}^K \alpha_{ij}^{(n)} g_{j}^{(l_i)}\left( u_{i}^{(n)} \right)} \mathbf{FFN}_{j}^{(l_i)} \left( u_{i}^{(n)} \right).
    \end{equation}
    \item \textit{Result Feedback}: Steps 2-5 are repeated for $L$ rounds until all layers are traversed. The final results are sent back to the user devices.

\end{enumerate}

\section{Problem Formulation}

In this section, we address two key problems to minimize the energy consumption in the DMoE protocol. We first formulate the expert selection problem, focusing solely on optimizing the expert allocation variables $\alpha_{ij}^{(n)}$. Furthermore, we extend to the JESA problem, aiming to derive optimal solutions for both the expert allocation $\alpha_{ij}^{(n)}$ and the subcarrier allocation $\beta_{ij}^{(m)}$.

\subsection{Expert Selection Problem Formulation}

\begin{figure}[t]
    \centering
    \includegraphics[width=0.8\linewidth]{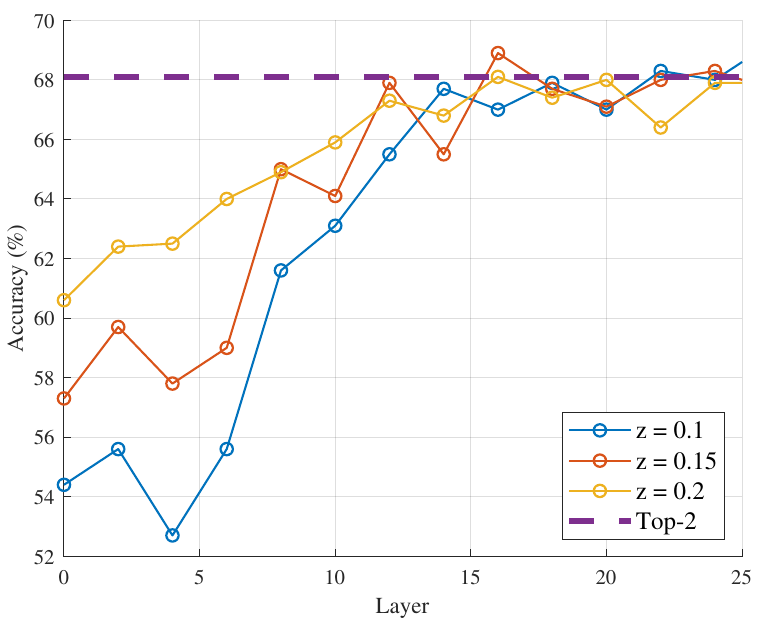}
    \caption{The final accuracy versus the starting layer index with lowered QoS requirement.}
    \label{fig:layer}
\end{figure}

We first propose a method for estimating task-relevance in DMoE. Since a higher gating score indicates better performance, for one $u_{i}^{(n)}$, the task-relevance of an expert combination can be estimated by the sum of gating scores: $\sum_{j=1}^K \alpha_{ij}^{(n)} \cdot g_{j}^{(l_i)} \left( u_{i}^{(n)} \right)$.
For the importance of gating scores across layers, extensive experiments have shown that the layers contribute differently to the final performance. Lower layers are more critical for accuracy, while upper layers allow greater flexibility.
The experiment is conducted by adjusting QoS thresholds across layers, which is to ensure the sum of gating scores meets or exceeds a specified threshold $z$. To determine the layer importance, we set a lower $z$ in 4 consecutive layers starting from a layer, while maintaining $z = 0.5$ in the remaining layers. Results in Fig.~\ref{fig:layer} validate our assumption.
To quantify this property, we introduce an importance factor $\gamma^{(l)}$ for each layer, assuming it is non-increasing: $\gamma^{(l)} \geq \gamma^{(l+1)}, \forall l$.
Assume that the QoS requirement of the query is that the expected performance of DMoE on each token must meet or exceed $z$ multiplies the layer importance factor $\gamma^{(l_i)}$, which can be expressed as
\begin{equation*}
    (\text{C1}) \quad \sum_{j=1}^K \alpha_{ij}^{(n)} g_{j}^{(l_i)} \left( u_{i}^{(n)} \right) \geq z \gamma^{(l_i)}.
\end{equation*}
Previous research and practices on MoE~\cite{fedus2022switch,lepikhin2021gshard} demonstrate that the performance does not increase monotonically with the number of experts selected. Instead, the performance increases initially with the rise of the number of experts but decreases afterward. Hence, we propose a constraint on the maximum number of experts as
\begin{equation*}
    (\text{C2}) \quad \sum_{j=1}^K \alpha_{ij}^{(n)} \leq D.
\end{equation*}
For the expert selection problem, expertise diversity is included in C1 while channel diversity is included in $R_{ij}$. We aim to minimize the total energy consumption from both communication and computation. Considering the expert allocation $\alpha_{ij}^{(n)}$ and assuming $R_{ij}$ in (\ref{eq:ecomm}) is known, we formulate the expert selection problem with QoS requirements and the expert number constraint, which can be expressed as:
\begin{equation*}
    (\text{P1}) \quad 
    \begin{alignedat}{2}
        \min_{\alpha_{ij}^{(n)}} \quad & \left( \sum_{i=1}^K \sum_{j=1, j \ne i}^K E_{ij}^{comm} + \right. && \left. \sum_{j=1}^K E_{j}^{comp} \right)\\
        \text{s.t.} \quad & \text{C1, C2} && \forall i,n,\\
        &\alpha_{ij}^{(n)} \in \{0,1\} && \forall i,j,n.
    \end{alignedat}
\end{equation*}



\begin{Remark}
    \label{remark}
    \emph{Certain $u_{i}^{(n)}$ may not satisfy C1 and C2 simultaneously. This occurs when the sum of the Top-$D$ experts cannot meet the required score $z \gamma^{(l_i)}$, rendering the problem infeasible. In such cases, we let them select the Top-$D$ experts.}
\end{Remark}


\subsection{Joint Expert and Subcarrier Allocation Problem Formulation}

For the ultimate problem of JESA, our goal is to select experts for hidden states and also allocate the appropriate subcarriers to facilitate these transmissions, to minimize the overall energy consumption from both transmission and computation in the DMoE system. Recalling the subcarrier allocation indicator $\beta_{ij}^{(m)}$, the exclusive subcarrier allocation constraint can be expressed as
\begin{equation*}
    (\text{C3}) \quad \sum_{i=1}^K \sum_{j=1, j \ne i}^K \beta_{ij}^{(m)} \leq 1.
\end{equation*}
The JESA problem is formulated under the constraints in P1 and the subcarrier availability constraint, as follows:
\begin{equation*}
(\text{P2}) \quad
    \begin{alignedat}{2}
    \min_{\alpha_{ij}^{(n)}, \beta_{ij}^{(m)}} & \left( \sum_{i=1}^K \sum_{j=1, j \ne i}^K E_{ij}^{comm} + \right. && \left. \sum_{j=1}^K E_{j}^{comp} \right)\\
    \text{s.t.} \quad &\text{C1, C2} && \forall i,n, \\
    &\text{C3} && \forall m,\\
    &\alpha_{ij}^{(n)}, \beta_{ij}^{(m)} \in \{0,1\} && \forall i,j,n,m.
    \end{alignedat}
\end{equation*}

\section{Optimal Expert Selection}

In this section, we address the fusion of expertise and channel diversity in expert selection. We first transform Problem P1 into a simplified one and then propose the optimal expert selection algorithm DES. DES features a linear relaxed lower bound as the bounding criterion for tree search, where the experts are excluded greedily based on the energy-to-score ratio. The bounding criterion significantly reduces the complexity of the problem.

\subsection{Problem Reformulation}

In P1, summations over $i$, $j$ and $n$ are involved. To simplify the expression, let $e_{ij} = s_0 \left( a_j + P_0 \sum_{m=1}^M \beta_{ij}^{(m)} / R_{ij} \right)$ for $i \ne j$ and $e_{jj} = s_0 a_j$. Note that
\begin{equation}
    \min_{\alpha_{ij}^{(n)}} \sum_{i=1}^K \sum_{j=1}^K \sum_{n=1}^{N} e_{ij} \alpha_{ij}^{(n)} \geq
    \sum_{i=1}^K \sum_{n=1}^{N} \left( \min_{\alpha_{ij}^{(n)}} \sum_{j=1}^K e_{ij} \alpha_{ij}^{(n)} \right),
\end{equation}
where equality holds if and only if all terms $\sum_{j=1}^K e_{ij} \alpha_{ij}^{(n)}$ reach their minimum values simultaneously, which can be achieved since $\alpha_{ij}^{(n)}$ for different $i$ and $n$ are independent. Thus, we can reformulate the problem as a single minimization problem for a given $i$ and $n$:
\begin{equation*}
(\text{P1(a)}) \quad
    \begin{aligned}
    \min_{\alpha_{ij}^{(n)}} \quad & \sum_{j=1}^K e_{ij} \alpha_{ij}^{(n)} \\
    \text{s.t.} \quad & \text{C1, C2}, \\
    &\alpha_{ij}^{(n)} \in \{0,1\} \quad \forall j.
    \end{aligned}
\end{equation*}
Problem P1 can be viewed as solving P1(a) for all possible $i$ and $n$. However, the combinatorial nature of the problem makes finding the optimal solution challenging. The following proposition, proved in Appendix~\ref{appx:nphard} by reducing the knapsack problem to P1(a), suggests that no known polynomial-time algorithm exists unless $\text{P} = \text{NP}$.

\begin{Proposition}
\label{prop:nphard}
\emph{Problem P1(a) is NP-hard.}
\end{Proposition}

\begin{proof}
See Appendix~\ref{appx:nphard}.
\end{proof}

\subsection{Search Tree Construction and Breadth-First Search}

Given the NP-hardness of expert selection, directly searching for the optimal solution incurs an exponential time complexity of $O \left( 2^K \right)$, which makes it intractable with a large number of experts like DeepSeek-V3 with $K = 256$~\cite{liu2024deepseek}.
We propose a time-efficient method by representing the solution space of the P1(a) as a binary search tree.
A node $v$ contains five elements $(j,t,e,P_\text{exc},P_\text{inc})$. $j$ is the next expert to be considered to be included or excluded, $t$ is the accumulated score, $e$ is the current energy, and $P_\text{exc}$ and $P_\text{inc}$ are two sets storing the past excluded and included experts respectively.
At each level of the tree, a decision is made for one expert: to exclude it in the final solution or not. Let there be $K$ experts, each with a score $t_j$ and energy $e_j$, and the QoS requirement is $z \gamma^{(l)}$. To construct the search tree, we create the root representing the initial state with no experts explicitly considered, thus $t = \sum_{j=1}^K g_{j}^{(l)} \left( u \right) = 1$ and $e = \sum_{j=1}^K e_j$. $P_\text{exc}$ and $P_\text{inc}$ are empty sets. The next expert is $j=1$, hence $v_0 = (1, 1, \sum_{j=1}^K e_j, \emptyset, \emptyset)$. For branching, at each node corresponding to the $j$-th expert, two child nodes are generated. The left child excludes expert $j$, and the total score and energy are updated: $t \leftarrow t - t_j$, $e \leftarrow e - e_j$. The right child includes expert $j$, thus the total value and weight remain unchanged: $t \leftarrow t$, $e \leftarrow e$.




To systematically explore the search tree, we employ a breadth-first search (BFS) traversal strategy.
To determine whether the child nodes are feasible, two criteria corresponding to C1 and C2 in P1(a) are employed. Child nodes are checked if their score $t \geq z \gamma^{(l)}$ and the number of included experts $|P_\text{inc}| \leq D$ before enqueueing. Then the BFS traversal proceeds as follows.
For initialization, BFS starts with the root node in the queue and sets the global minimum energy $e_\text{min} \leftarrow \sum_{j=1}^K e_j$ and the corresponding excluded expert set $P_\text{min} \leftarrow \emptyset$. While the queue is not empty, the following steps are executed. First, dequeue the front node and get $(j,t,e,P_\text{exc},P_\text{inc})$. Second, compare its total energy $e$ with the best-known $e_\text{min}$. If better, update $e_\text{min}$ and $P_\text{min}$. Third, generate its child nodes to exclude and include the next expert and check if $t \geq z \gamma^{(l)}$ and $|P_\text{inc}| \leq D$. At last, enqueue the feasible child nodes.

\subsection{Bounding Criterion}

We introduce an additional bounding criterion to optimize tree search by calculating the energy lower bound for further nodes.
In our approach, we relax the binary constraint on $\alpha_{ij}^{(n)}$ and C2 of P1(a), transforming the problem into a linear programming (LP) formulation, expressed as:
\begin{equation*}
(\text{P1(b)}) \quad
    \begin{alignedat}{2}
    \min_{x_j} \quad & \sum_{j=j'}^K e_{j} x_j \\
    \text{s.t. } \quad & \sum_{j=j'}^K t_j x_j \geq z \gamma^{(l)}, \\
    & x_j \in [0,1] && \forall j.
    \end{alignedat}
\end{equation*}
Introducing $\lambda$ as the Lagrange multiplier, P1(b) can be reformulated as:
\begin{equation*}
(\text{P1(c)}) \quad
    \begin{alignedat}{2}
    \min_{x_j, \lambda} \quad & f\left( x_j, \lambda \right) = \sum_{j=j'}^K e_{j} x_j && + \lambda \left( \sum_{j=j'}^K t_j x_j - z \gamma^{(l)} \right) \\
    \text{s.t.} \quad & x_j \in [0,1] && \forall j.
    \end{alignedat}
\end{equation*}
The optimal $x_j^*$ is found by solving $\left. \frac{\partial f}{\partial x_j} \right|_{x_j = x_j^*} = 0$, where the partial derivative is:
\begin{equation}
    \left. \frac{\partial f}{\partial x_j} \right|_{x_j = x_j^*} = e_j + \lambda t_j
    \begin{cases}
        > 0, & \text{if } \lambda > - e_j / t_j, \\
        = 0, & \text{if } \lambda = - e_j / t_j, \\
        < 0, & \text{if } \lambda < - e_j / t_j.
    \end{cases}
\end{equation}

Since the derivative is independent of $x_j$, the optimal value either occurs at the boundaries of $\{0,1\}$ in the feasible region or the derivative equals $0$, and hence the optimal solution is obtained inside the feasible region. The optimal value of $x_j$ depends on the relationship between $\lambda$ and $- e_j / t_j$, which represents the energy-to-score ratio of each expert. To meet the QoS requirement in P1(b), $\lambda$ will correspond to a critical expert: including this expert satisfies the QoS, but excluding it does not. Then the critical expert $j_\text{crit}$ can be identified by greedily excluding experts until the QoS limit is reached.
Assuming the experts are sorted in descending energy-to-score ratio, the optimal solution is
\begin{equation}
x_j = \begin{cases}
1, & \text{if } j < j_\text{crit}, \\
\left(z \gamma^{(l)} - \sum_{j=j'}^{j_\text{crit} - 1} t_j \right) / t_j, & \text{if } j = j_\text{crit}, \\
0, & \text{if } j > j_\text{crit}.
\end{cases}
\label{eq:optimal_x}
\end{equation}

By greedily excluding experts until the QoS limit is reached, we derive a lower bound. If this bound exceeds the current minimum energy found, it indicates that further search will not yield a better solution. The bounding process includes:
\begin{enumerate}
    \item \textbf{Sort Experts}: Experts are sorted in descending order of energy-to-score ratio $e_j / t_j$;
    \item \textbf{Compute Bound}: Starting from the current node, experts are added fractionally to maximize total value without being lower than the QoS requirement $z$. The lower bound $e_\text{bound}$ is calculated as:
\begin{equation}
   e_\text{bound}  = e - \sum_{j=j'}^K e_j x_j,
\end{equation}

where $x_j$ is the fraction of expert $j$ included, defined in (\ref{eq:optimal_x});

    \item \textbf{Pruning Decision}: After computing the upper bound, the node is evaluated:

\begin{itemize}
    \item \textbf{Continue Exploration}: If $e_\text{bound} < e_\text{min}$, the node is promising, and its child nodes are enqueued for exploration.
    \item \textbf{Prune Node}: If $e_\text{bound} \geq e_\text{min}$, the node is not promising, and it, along with its descendants, is discarded.
\end{itemize}

\end{enumerate}
This pruning significantly reduces the number of nodes to be explored, enhancing the algorithm's efficiency. 
The detailed implementation of optimal expert selection is provided in Algorithm~\ref{algo:expert}.


\begin{algorithm}[t]
\caption{Optimal Expert Selection (DES)}
\label{algo:expert}
\textbf{Input:} The maximum achievable data rate $R_{ij}$, the gating scores $g_{j}^{(l_i)} \left( u_i^{(n)} \right)$, the QoS requirement $z \gamma^{(l)}$, the maximum number of experts $D$, the computation energy $a_j$;\\
\textbf{Initialize:} $\alpha_{ij}^{(n)} \leftarrow 0, Q \leftarrow \text{empty queue}, t_j \leftarrow g_{j}^{(l_i)} \left( u_i^{(n)} \right), e_j \leftarrow a_j + P_0 s_0 / R_{ij}, e_{\text{min}} \leftarrow \sum_{j=1}^K e_j, P_{\text{min}} \leftarrow \emptyset$; \\


Sort experts by $e_j / t_j$ in descending order; \\
$Q. \text{enqueue} ((1, 1, e_{\text{min}}, \emptyset, \emptyset))$; \\

\textbf{while} $Q \ne \emptyset$ \textbf{do}\\
\quad $(j, t, e, P_{\text{exc}}, P_{\text{inc}}) \leftarrow Q$.dequeue(); \\

\quad \textbf{if} $w \geq z \gamma^{(l)}$ \textbf{and} $e < e_{\text{min}}$ \textbf{and} $|P_{\text{exc}}| \geq K - D$ \textbf{then} \\
\quad \quad $e_{\text{min}} \leftarrow e, P_{\text{min}} \leftarrow P_{\text{exc}}$; \\

\quad \textbf{if} $w < z \gamma^{(l)}$ \textbf{or} $j > K$ \textbf{then} \\
\quad \quad \textbf{continue} \\

\quad $e_{\text{bound}} \leftarrow \text{bound} ((j, t, e), z, \mathbf{t}, \mathbf{e})$; \\

\quad \textbf{if} $e_{\text{bound}} \leq e_{\text{min}}$ \textbf{then} \\
\quad \quad $Q. \text{enqueue} ((j + 1, t - t_j, e - e_j, P_{\text{exc}} \cup \{ j \}, P_{\text{inc}}))$; \\

\quad \quad \textbf{if} $|P_{\text{inc}}| < D$ \textbf{then} \\
\quad \quad \quad $Q. \text{enqueue} ((j + 1, t, e, P_{\text{exc}}, P_{\text{inc}} \cup \{ j \}))$; \\

Retrieve the excluded experts from $P_{\text{min}}$, set $\alpha_{ij}^{(n)} \leftarrow 0$, and set $\alpha_{ij}^{(n)} \leftarrow 1$ for the remaining experts;\\

\textbf{return} $\alpha_{ij}^{(n)}$\\

\textbf{function} bound$((j, t, e), z, \mathbf{t}, \mathbf{e})$\\
\quad \textbf{if} $t \leq z$ \textbf{then} \\
\quad \quad \textbf{return} $0$ \\

\quad \textbf{while} $j \leq K$ \textbf{and} $s - s_j \geq z$ \textbf{do} \\

\quad \quad $t \leftarrow t - t_j, e \leftarrow e - e_j$; \\
\quad \quad $j \leftarrow j + 1$; \\

\quad \textbf{if} $j \leq K$ \textbf{then}\\
\quad \quad $e \leftarrow e - (z - t) e_j / t_j$; \\

\textbf{return} $e$ \\
\end{algorithm}

\section{Joint Expert and Subcarrier Allocation}

In expert selection, we get insights from handling expertise and channel diversity integration. In this section, we extend the problem to JESA, where the expert selection and subcarrier allocation are jointly considered to minimize the overhead. We first derive the optimal subcarrier allocation algorithm and then integrate it with the expert selection algorithm to solve the JESA problem. A special structure of the JESA problem is exploited to ensure the near-optimal performance of the proposed algorithm.

\subsection{Optimal Subcarrier Allocation}

We derive a sub-problem of P1 for subcarrier allocation determining $\beta_{ij}^{(m)}$ after the expert selection $\alpha_{ij}^{(n)}$ has been decided. Since the computation energy is fixed once the expert allocation is known, our goal turns to minimize the communication energy only, which can be expressed as:
\begin{equation*}
(\text{P3}) \quad
\begin{alignedat}{2}
    \min_{\beta_{ij}^{(m)}} \quad & \sum_{i=1}^K \sum_{j=1, j \ne i}^K E_{ij}^{comm} \\
    \text{s.t.} \quad & \text{C3} && \forall m,\\
    &\beta_{ij}^{(m)} \in \{0,1\} && \forall i,j,m.
\end{alignedat}
\end{equation*}
Appendix~\ref{appx:sa} shows P3 can be transformed into an assignment problem. Several algorithms can be applied to obtain the optimal solution~\cite{wang2020convergence,mecsurvey}. When using the Kuhn-Munkres algorithm with Fibonacci heaps~\cite{kuhn1955hungarian}, the time complexity is $O(M^2K(K-1) + M^2 \log M)$, and for a balanced assignment the complexity is $O(M^3)$.

\subsection{JESA via Block Coordinate Descent}

We have already derived the optimal solutions for two lower-level sub-problems, P1 for expert selection and P3 for subcarrier allocation. To address the upper-level JESA Problem P2, we employ the block coordinate descent (BCD) approach.
BCD is a well-established iterative approach to exactly or approximately solve a large number of optimization problems by using a series of simple, coordinate-wise updates. While BCD is not generally guaranteed to converge to the global optimum, we identify a unique structure of the JESA problem that ensures BCD achieves near-optimal performance. This is formalized in the following theorem:

\begin{Theorem}
\label{theo:optim}
\emph{Let the optimal solution of P2 be $\alpha^*$ and $\beta^*$, and the solution produced by Algorithm~\ref{algo:bcd} be $\alpha$ and $\beta$. Assume $r_{ij}^{(m)}$ for different $i,j,m$ are independent and identically distributed (i.i.d.). Then the probability that the algorithm finds the optimal solution is bounded by
\begin{equation}
\Pr \left( \alpha = \alpha^*, \beta = \beta^* \right) \geq \frac{\prod_{i=0}^{K(K-1)-1}(M-i)}{M^{K(K-1)}}.
\end{equation}}
\end{Theorem}

\begin{proof}
Since $r_{ij}^{(m)}$ are i.i.d., the maximum value of $r_{ij}^{(m)}$ can occur for different $m$ with equal probability. Let $A$ denote the event where the maximum values of $r_{ij}^{(m)}$ for all $K(K-1)$ links occur on distinct subcarriers. The probability of this event is
\begin{equation}
\Pr (A) = \frac{\prod_{i=0}^{K(K-1)-1}(M-i)}{M^{K(K-1)}}.
\end{equation}
Note that if this event occurs, each link choosing the subcarrier with the largest $r_{ij}^{(m)}$ is the optimal solution, independent of the expert allocation. Hence, the subcarrier allocation algorithm within Algorithm~\ref{algo:bcd} will yield the optimal $\beta^*$, and given $\beta^*$, Algorithm~\ref{algo:expert} will also produce the optimal $\alpha^*$. Since $A$ is a necessary condition, it follows that
\begin{equation}
    \Pr \left( \alpha = \alpha^*, \beta = \beta^* \right) \geq \frac{\prod_{i=0}^{K(K-1)-1}(M-i)}{M^{K(K-1)}}.
\end{equation}
This completes the proof of Theorem~\ref{theo:optim}.
\end{proof}

Further noting that $\Pr \left( \alpha = \alpha^*, \beta = \beta^* \right) \rightarrow 1$ as $M \rightarrow \infty$, we can derive the following corollary:
\begin{Corollary}[Asymptotically Optimal]
\emph{When $M$ is sufficiently large, the solution derived by Algorithm~\ref{algo:bcd} will be optimal.}
\end{Corollary}

\begin{Remark}
\emph{The assumption of a large $M$ is realistic in modern communication systems. For example, 802.11ax~\cite{80211ax} supports up to 2048 subcarriers, and 5G NR~\cite{3gpp_nr} can accommodate several thousand subcarriers. With $K=4$ and $M=2048$, the probability of obtaining the optimal solution exceeds 96.8\%.}
\end{Remark}

For Problem P2, we have noted that the constraints on $\alpha_{ij}^{(n)}$ and $\beta_{ij}^{(m)}$ are defined separately. When optimizing $\alpha_{ij}^{(n)}$ with $\beta_{ij}^{(m)}$ fixed, we only need to consider the C1 and C2, reducing P2 to P1, where we can apply Algorithm~\ref{algo:expert} to obtain an optimal solution. Similarly, when optimizing $\beta_{ij}^{(m)}$, P2 is reduced to P3 and we can use the subcarrier allocation algorithm. This allows us to sequentially minimize the energy consumption in P2 by optimizing each variable block in a round-robin fashion. We begin by initializing $\alpha_{ij}^{(n)}$ and $\beta_{ij}^{(m)}$ with a feasible solution, and then alternatively optimize $\alpha_{ij}^{(n)}$ and $\beta_{ij}^{(m)}$ until convergence. The implementation of solving the expert and subcarrier allocation problem P2 using BCD is outlined in Algorithm~\ref{algo:bcd}.

\begin{algorithm}[t]
\caption{Joint Expert and Subcarrier Allocation (JESA)}
\label{algo:bcd}
\textbf{Input:} The gating scores $g_{j}^{(l_i)} \left( u_{i}^{(n)} \right)$, the QoS requirement $z \gamma^{(l_i)}$, the maximum number of experts $D$, the computation power coefficient $a_j$, the data rate of subcarriers $r_{ij}^{(m)}$;\\
\textbf{Initialize:} $\alpha_{ij}^{(n)} \leftarrow 1, \beta_{ij}^{(m)} \leftarrow \text{Random Assign}(M, K(K-1))$; \\

\textbf{repeat:} \\

\quad $R_{ij} \leftarrow \sum_{m=1}^M \beta_{ij}^{(m)} r_{ij}^{(m)}$; \\
\quad \textbf{for} $i = 1, 2, \ldots, K$ \textbf{do}\\
\quad \quad \textbf{for} $n = 1, 2, \ldots, N_i$ \textbf{do}\\
\quad \quad \quad $\alpha_{ij}^{(n)'} \leftarrow$ DES$\left( \beta_{ij}^{(m)}, g_{j}^{(l_i)} \left( u_{i}^{(n)} \right), z \gamma^{(l_i)}, a_j \right)$;\\

\quad \textbf{for} $i = 1, 2, \ldots, K$ \textbf{do} \\
\quad \quad \textbf{for} $j = 1, 2, \ldots, K$ \textbf{do} \\
\quad \quad \quad $s_{ij} \leftarrow \sum_{n=1}^{N_i} \alpha_{ij}^{(n)'}$; \\
\quad $\beta_{ij}^{(m)'} \leftarrow$ Subcarrier Allocation$\left( s_{ij}, r_{ij}^{(m)} \right)$;\\

\textbf{until} $\alpha_{ij}^{(n)'} = \alpha_{ij}^{(n)}$ \textbf{and} $\beta_{ij}^{(m)'} = \beta_{ij}^{(m)}$\\

\textbf{return} $\alpha_{ij}^{(n)}, \beta_{ij}^{(m)}$
\end{algorithm}


The following proposition ensures the feasibility of the intermediate variables $\alpha_{ij}^{(n)}$ and $\beta_{ij}^{(m)}$ for P2, guaranteeing progress in each optimization step and the overall convergence of the algorithm. Moreover, since $\alpha_{ij}^{(n)}$ is updated to its optimal value given $\beta_{ij}^{(m)}$, and vice versa, the BCD algorithm will converge within a few iterations, thereby reducing the algorithm's running time.

\begin{Proposition}
\label{prop:bcd}
\emph{In each optimization step, the intermediate $\alpha_{ij}^{(n)}$ and $\beta_{ij}^{(m)}$ are feasible and conditionally optimal for P2, and Algorithm~\ref{algo:bcd} makes progress with each iteration until convergence.}
\end{Proposition}

\begin{proof}
    See Appendix~\ref{appx:bcd}.
\end{proof}


\section{Experimental Results}

\subsection{Experimental Settings}

\subsubsection{Models and Datasets}

\begin{itemize}
    \item \textit{Expert Selection}: The experts deployed on edge nodes are Llama-3-8B-Instruct~\cite{llama3instruct} for general knowledge, Llama3-8B-Chinese-Chat~\cite{llama3chinese} specialized in Chinese Q\&A and Llama3-OpenBioLLM-8B~\cite{llama3bio} specialized in biomedical tasks. Gates are prepared according to their specialized domains. Models are tested on multi-domain datasets, including MMLU~\cite{hendrycks2021mmlu} for world knowledge, C-Eval~\cite{huang2024ceval} and CMMLU for Chinese Q\&A, and MMLU-Bio~\cite{hendrycks2021mmlu} and MedMCQA~\cite{pmlr-v174-pal22a-medmcqa} for medical queries;
    \item \textit{Energy Efficiency}: We initialize a wireless network with $K = 8$ devices. Using Mixtral-8x7B-Instruct-v0.1~\cite{jiang2024mixtral} as the MoE model, the DMoE system is initialized as Section~\ref{sec:initial}. The gates are pre-configured. The test dataset is MMLU-Anatomy, with each user assigned 1/8 of the data.
\end{itemize}

\subsubsection{Devices and Communications}

We set bandwidth to $B_0 = 1$ MHz and the transmission power to $P_0 = 1 \times 10^{-2}$ W, with a signal-to-noise ratio (SNR) of $P_0 / N_0 = 10$ dB. The size of a hidden state is set to $s_0 = 8$ kB, as defined in Llama-3 and Mixtral. We assume the channel gain follows Rayleigh fading with an average path loss of $10^{-2}$. The computation power coefficient is set to $a_j = j \times 10^{-3}$ J/token.

\subsubsection{Benchmark Schemes}

We adopt the following benchmark schemes:

\begin{itemize}
    \item \textit{Top-k Allocation} (Top-k): Select $k$ experts with the highest scores, and then perform optimal subcarrier allocation;
    \item \textit{Homogeneous Allocation} (H($z, D$)): Set QoS requirement to be $z$, the maximum expert number $D$, and a homogeneous factor $\gamma^{(l)} = 1$ for all layers $l$. Use Algorithm~\ref{algo:bcd} to perform expert and subcarrier allocation;
    \item \textit{Joint Expert and Subcarrier Allocation} (JESA($\gamma_0, D$)): Set $z = 1$, $\gamma^{(l)} = \gamma_0^l$ and adjust $\gamma_0$ to control the layer-specific QoS requirements, followed by energy-efficient JESA Algorithm~\ref{algo:bcd};
    \item \textit{Lower Bound} (LB($\gamma_0, D$)): Utilize the DES Algorithm~\ref{algo:expert} with $\gamma^{(l)} = \gamma_0^l$, but ignore the exclusive subcarrier allocation constraint in P3 (i.e., allowing concurrent subcarrier occupation) and select the subcarrier with the highest rate to obtain a lower bound on energy consumption.
\end{itemize}

\subsection{Expert Selection of DMoE}

We begin by examining the expert selection pattern of the DES Algorithm~\ref{algo:expert}, focusing on the tradeoff between expertise and channel diversity in expert selection. We manually create high-performing experts with higher gating scores and set their power consumption to be proportionally higher. Meanwhile, we include some low-performing, low-cost experts. The expert selection patterns of different values of $\gamma_0$ are shown in Fig.~\ref{fig:expert_selection}, where deeper color indicates higher selection probabilities. As expected, the results show that at lower layers, the task-relevance requirement is higher, thus DES favors high-performing experts to ensure accurate results, despite their higher power consumption. At higher layers, however, DES shifts towards low-cost experts to reduce power consumption, resulting in a noticeable change in the selection pattern. Furthermore, the value of $\gamma_0$ influences the point at which this shift occurs; a larger $\gamma_0$ delays the shift, allowing for the selection of more high-performing experts, leading to better inference results with higher power costs.

\begin{figure}[t]
    \centering
    \includegraphics[width=0.85\linewidth]{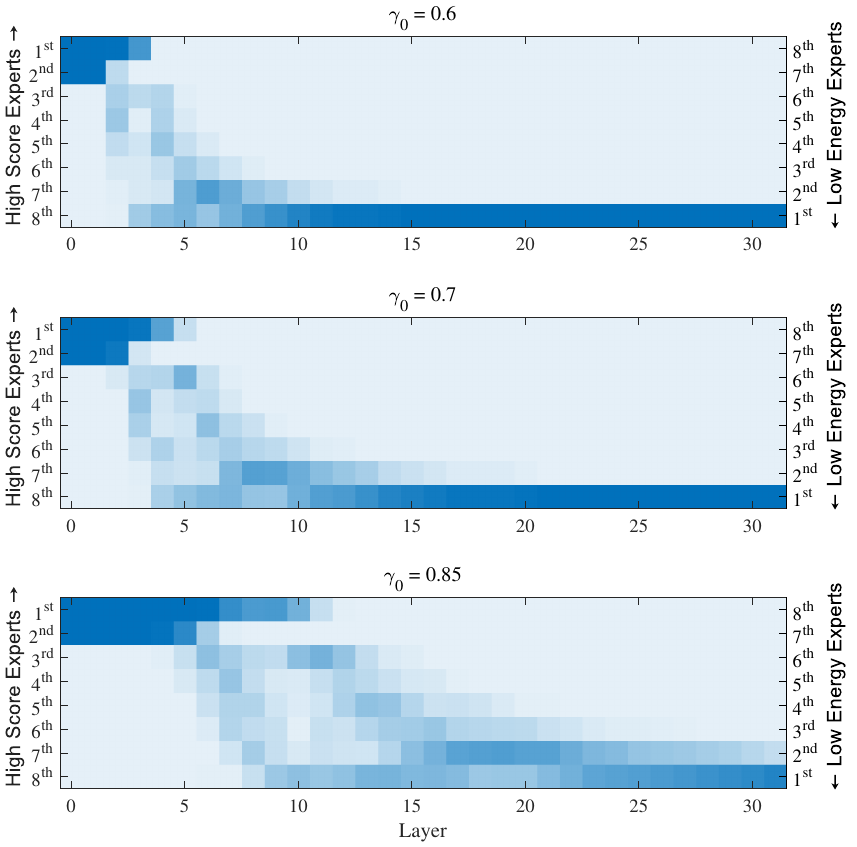}
    \caption{Expert selection patterns with layer importance factor $\gamma^{(l)} = \gamma_0^l$. A deeper color indicates a higher selection probability. At higher layers, DES shifts from preferring high-performing experts to low-cost experts.}
    \label{fig:expert_selection}
\end{figure}

Table~\ref{tab:result_llama} presents the inference accuracy and relative energy consumption (normalized to 1.0) of DES Algorithm~\ref{algo:expert} across multi-domain tasks.
Compared to the static Top-k selection which is simple but comes with huge overhead, the DES balances the task-relevance and energy consumption. Taking both into consideration, DES maintains high accuracy while significantly reducing energy consumption.
By setting a higher $\gamma_0$ in the layer importance factor $\gamma^{(l)} = \gamma_0^l$, DES achieves better performance across all tasks, albeit with increased power costs. Overall, the results demonstrate that DMoE protocol effectively leverages the specialized knowledge of experts and minimizes the overhead, resulting in high task-relevance expert selection, while dynamic expert selection significantly reduces power consumption, enabling energy-efficient inference.

\begin{table*}[t]
\caption{Performance of Dynamic Expert Selection on Multi-Domain Tasks}
\centering
\begin{tabular}{lcccccccccc}
\toprule
Model & \multicolumn{2}{c}{MMLU} & \multicolumn{2}{c}{C-Eval} & \multicolumn{2}{c}{CMMLU} & \multicolumn{2}{c}{MMLU-Bio} & \multicolumn{2}{c}{MedMCQA}\\
\cmidrule(lr){2-3} \cmidrule(lr){4-5} \cmidrule(lr){6-7} \cmidrule(lr){8-9} \cmidrule(lr){10-11} 
& Acc & En & Acc & En & Acc & En & Acc & En & Acc & En\\
\midrule
\textit{Individual Experts}\\
\quad Llama3-8B-Instruct & 63.8 & - & 51.4 & - & 51.2 & - & 72.3 & - & 57.0 & -\\
\quad Llama3-8B-Chinese-Chat & 63.1 & - & 51.4 & - & 52.1 & - & 72.2 & - & 55.3 & -\\
\quad Llama3-OpenBioLLM-8B & 61.1 & - & 48.0 & - & 47.3 & - & 76.2 & - & 57.7 & -\\
\midrule
\textit{Conventional Expert Selection}\\
\quad Llama3-MoE-3x8B (Top-1) & 63.6 & 0.48 & 51.4 & 0.51 & 51.2 & 0.46 & 73.7 & 0.46 & 57.6 & 0.47\\
\quad Llama3-MoE-3x8B (Top-2) & 64.1 & 1.00 & 51.9 & 1.00 & 51.9 & 1.00 & 75.5 & 1.00 & 58.4 & 1.00\\
\midrule
\textit{Dynamic Expert Selection (Proposed)}\\
\quad Llama3-MoE-3x8B (DES(0.6, 2)) & 64.2 & 0.12 & 51.9 & 0.19 & 51.6 & 0.12 & 73.1 & 0.14 & 57.4 & 0.14\\
\quad Llama3-MoE-3x8B (DES(0.7, 2)) & 64.4 & 0.16 & 52.1 & 0.24 & 51.6 & 0.14 & 73.7 & 0.22 & 57.7 & 0.20\\
\quad Llama3-MoE-3x8B (DES(0.8, 2)) & 64.5 & 0.19 & 52.3 & 0.30 & 51.8 & 0.24 & 73.6 & 0.26 & 58.0 & 0.23\\
\bottomrule
\end{tabular}
\label{tab:result_llama}
\end{table*}

\subsection{Energy Efficiency of DMoE}

We explore the energy-saving effects of JESA Algorithm~\ref{algo:bcd} at different layers of DMoE in Fig.~\ref{fig:energy}. It is worth noting that Algorithm~\ref{algo:bcd} achieves a near-optimal performance for all $\gamma_0$. One can observe that the cost per token of Top-2 scheduling is steady for all layers and fluctuates around a certain value. The selection policy we proposed, however, is dynamic and tends to select low-cost experts in high layers and attain a large reduction of energy consumption. The higher the layer is, the less contribution it will make to the final result, thus we can select experts at a lower cost with less attention to the gating scores. Moreover, it is obvious that by controlling the importance factor $\gamma_0$, the curves exhibit downward trends at different rates. A smaller $\gamma_0$ relaxes the task-relevance requirement thus a looser expert selection can be exploited, and hence a more rapid downward can be expected. Nevertheless, setting $\gamma_0$ too small is not recommended as it may incorporate too many low-performing experts in the lower layers, which could lead to a substantial accuracy decrease.

\begin{figure}[t]
    \centering
    \includegraphics[width=0.8\linewidth]{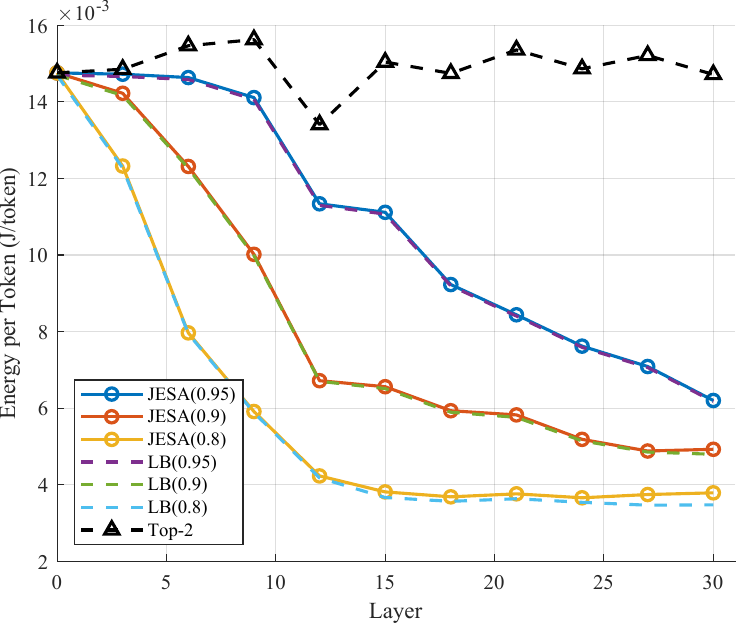}
    \caption{Energy consumption per token at different layers. JESA significantly reduced the communication and computation overhead of DMoE at higher layers and is close to the lower bound (LB).}
    \label{fig:energy}
\end{figure}

We further analyze the communication and computation energy consumed at different layers, as shown in Fig.~\ref{fig:comm_en} and Fig.~\ref{fig:comp_en} respectively, and compare the JESA with Top-2 and homogeneous allocation. JESA demonstrates a distinct approach to energy reduction compared to homogeneous allocation, in both communication and computation costs. Specifically, it maintains higher energy consumption at lower layers, with a gradual reduction as the scheduling progresses through higher layers. In contrast, homogeneous allocation preserves the behavior of Top-2 selection, reducing energy consumption at a similar rate across layers, regardless of their depth. Although both allocation schemes reduce energy consumption compared to Top-2 selection, our findings show that JESA results in lower average energy consumption for both communication and computation.

\begin{figure}[t]
    \centering
    \includegraphics[width=0.8\linewidth]{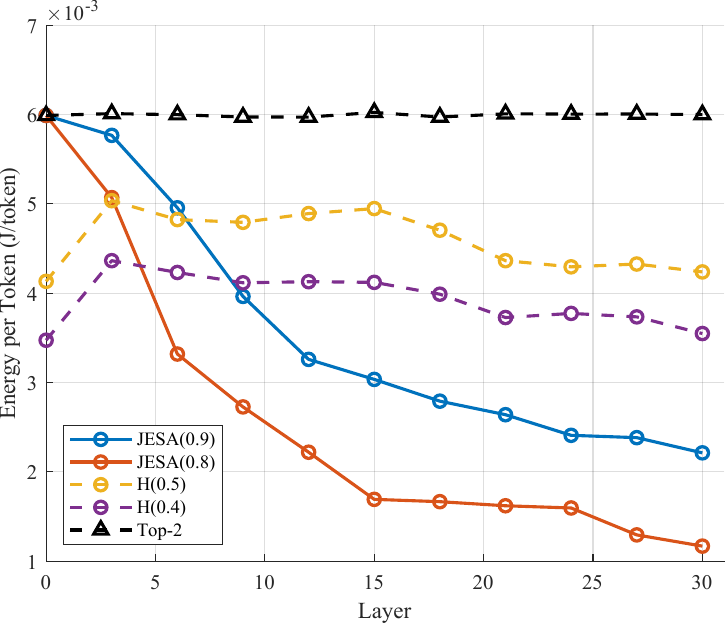}
    \caption{Communication energy consumption per token at different layers.}
    \label{fig:comm_en}
\end{figure}

\begin{figure}[t]
    \centering
    \includegraphics[width=0.8\linewidth]{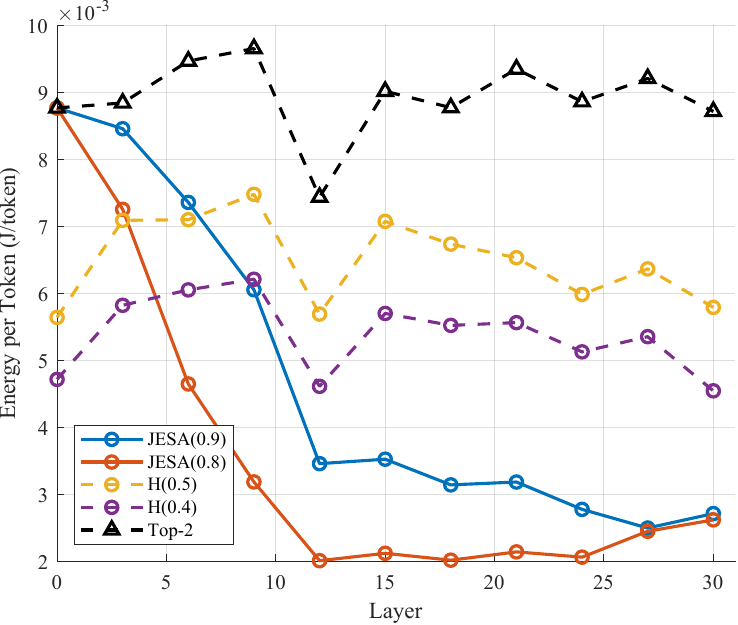}
    \caption{Computation energy consumption per token at different layers.}
    \label{fig:comp_en}
\end{figure}

Fig.~\ref{fig:en_acc} shows the tradeoff between average accuracy and total energy consumption in the DMoE system, tested on MMLU-Anatomy. With JESA, the overall energy consumed by communication and computation of the system can be reduced considerably to the level of 50\%, while the system still maintains a high accuracy. Using JESA and setting $\gamma_0 = 0.92$, a 40\% energy fall-off is achieved in exchange for only 5\% accuracy decrease compared to Top-2 scheduling. A greater energy consumption cut can also be achieved, though with higher accuracy sacrifice. Moreover, the advantage of JESA over homogeneous allocation can be evidently observed. The JESA algorithm steadily outperforms homogeneous policy, as JESA consumes much less energy to reach the same accuracy, and acquires higher accuracy for the same energy consumed. This advantage of heterogeneous policy over homogeneous one confirms the validity of our algorithm. It is worth mentioning that both the JESA and homogeneous selections converge to the Top-k selection for $\gamma \rightarrow 1$, which agrees with our algorithm design and Remark~\ref{remark}.

\begin{figure}[t]
    \centering
    \includegraphics[width=0.8\linewidth]{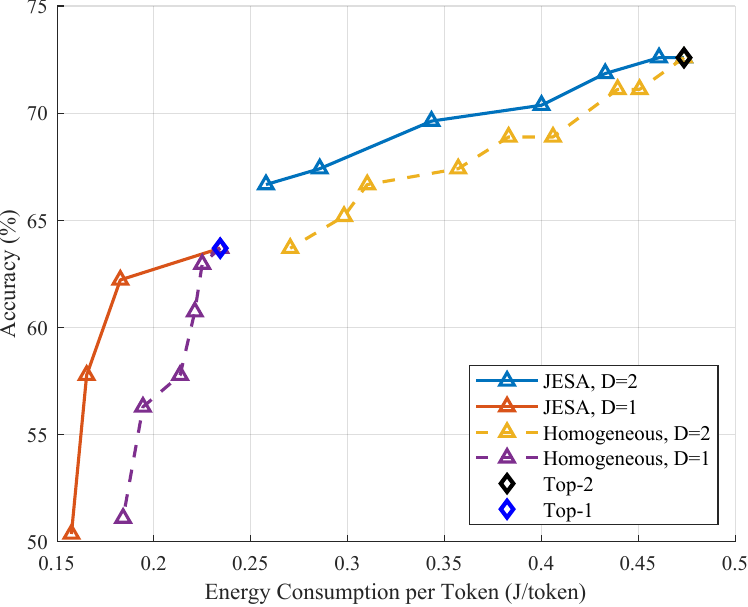}
    \caption{The tradeoff between inference accuracy and energy consumption. The DMoE protocol's overhead can be reduced considerably by JESA with a little decrease in accuracy. JESA always outperforms the homogeneous allocation.}
    \label{fig:en_acc}
\end{figure}

\section{Conclusion}

In this work, we proposed the DMoE protocol and addressed the expert selection and subcarrier allocation problem within the system. The DMoE framework represents a specialized instance of the general Integrated AI and Communications (IAAC). A key challenge in IAAC lies in managing the unprecedented integration of AI metrics with communication requirements. This challenge stands in contrast to traditional AI and communication systems, which focus solely on performance or transmission. The JESA algorithm effectively balances the tradeoff between task-relevance and channel conditions by jointly optimizing AI performance and RRM.
Notably, we demonstrated that introducing a hyper-parameter allows precise control over the tradeoff between performance and efficiency, enabling the framework to adapt to broader scenarios. The proposed framework minimizes energy consumption while maintaining high quality through task-relevant expert selection and efficient subcarrier allocation, which is a promising solution for connected intelligence.

This work introduces a novel direction in the field of edge AI and MoE, highlighting the paradigm shift in further research where joint analysis of the AI behaviors and communication nature is needed to optimize the system performance.
By integrating AI and communication into an end-to-end framework, this paradigm delivers substantial improvements in AI performance and resource utilization, paving the way for controllable and intelligent edge AI systems.
Several promising areas warrant further exploration. First, regarding the integration of AI and communication technologies, future research could focus on further optimizing the inference overhead by leveraging advanced 5G/6G features, such as ultra-reliable low-latency communications. Second, considering the wireless channel nature, the random participation of edge nodes incorporating the dynamic entrance and exit of experts could enable ad-hoc DMoE assembling. Lastly, for scenarios involving massive IoT device communication, future research directions include designing sparse routing and retrieval mechanisms to efficiently manage interactions with massive tiny experts (over a million).

\appendix

\subsection{Proof of Proposition~\ref{prop:nphard}}
\label{appx:nphard}
We prove Proposition~\ref{prop:nphard} by reducing the NP-hard knapsack problem to P1(a). The knapsack problem is defined as follows:
\begin{equation*}
(\text{KP}) \quad
    \begin{alignedat}{2}
    \max_{x_i} \quad & \sum_{i=1}^K v_i x_i \\
    \text{s.t.} \quad & \sum_{i=1}^K w_i x_i \leq c,\\
    &x_i \in \{0,1\} && \forall i.
    \end{alignedat}
\end{equation*}
Let $\xi_i = 1 - x_i$, so that $x_i = 1 - \xi_i$. Substituting $x_i$ with $\xi_i$ in the objective function and constraints, and converting the maximization into minimization, we obtain:
\begin{equation*}
(\text{KP(a)}) \quad
    \begin{alignedat}{2}
    \min_{\xi_i} \quad & \sum_{i=1}^K v_i \xi_i \\
    \text{s.t.} \quad & \sum_{i=1}^K w_i \xi_i \geq \sum_{i=1}^K w_i - c,\\
    &\xi_i \in \{0,1\} && \forall i.
    \end{alignedat}
\end{equation*}
This formulation is also NP-hard. Next, let $\xi_i = \alpha_{j}^{(n)}$, $v_i = e_{j}$, $w_i = g_j^{(l)} \left( u^{(n)} \right)$, and $c = z \gamma^{(l)}$. This transforms the problem into P1(a), except for the second constraint on the maximum number of experts. If we assume P1(a) is solvable in polynomial time for any $D$ in the second constraint, letting $D \rightarrow \infty$ would automatically satisfy this constraint, thereby providing a polynomial-time algorithm for KP(a), which contradicts the NP-hardness of KP(a). This completes the proof of Proposition~\ref{prop:nphard}.

\subsection{Optimal Subcarrier Allocation}
\label{appx:sa}


To solve P3, noting that for $\forall i,j$,
\begin{equation}
\frac{\sum_{m=1}^M \beta_{ij}^{(m)}}{\sum_{m=1}^M \beta_{ij}^{(m)} r_{ij}^{(m)}} \geq \frac{1}{\max_m{r_{ij}^{(m)}}},
\end{equation}
equality holds if and only if for $m' = \arg \max_m r_{ij}^{(m)}$, $\beta_{ij}^{(m')} = 1$ and $\beta_{ij}^{(m)} = 0$ otherwise. Leveraging this property, Problem P3 can be rewritten as:
\begin{equation*}
(\text{P3(a)}) \quad
    \begin{alignedat}{2}
    \min_{\beta_{ij}^{(m)}} \quad & \sum_{i=1}^K \sum_{j=1, j \ne i}^K \frac{P_0 s_{ij}}{\sum_{m=1}^M \beta_{ij}^{(m)} r_{ij}^{(m)}} \\
    \text{s.t.} \quad & \text{C3} && \forall m,\\
    &\sum_{m=1}^M \beta_{ij}^{(m)} \leq 1 && \forall i,j,\\
    &\beta_{ij}^{(m)} \in \{0,1\} && \forall i,j,m.
    \end{alignedat}
\end{equation*}
The problem can be represented using a bipartite graph, with one set of vertices representing the links from $i$ to $j$ and the other set representing the subcarriers.
The weight of edge connecting link $i$ to $j$ and subcarrier $m$ is the energy cost $w_{ij}^{(m)} = P_0 s_{ij} / r_{ij}^{(m)}$.
The optimization objective of P3(a) is to find a matching of links and subcarriers in the weighted bipartite graph that minimizes the sum of the edge weights.
This assignment problem is known to be solvable in polynomial time, and several assignment algorithms can be adapted to find the optimal subcarrier allocation~\cite{kuhn1955hungarian}.

\subsection{Proof of Proposition~\ref{prop:bcd}}
\label{appx:bcd}
We first prove the feasibility of the intermediate solution for P1 during the optimization. We initialize $\alpha_{ij}^{(n)}$ and $\beta_{ij}^{(m)}$ as a feasible solution for P1. Noting that the constraints on $\alpha_{ij}^{(n)}$ and $\beta_{ij}^{(m)}$ are separable in P1, and P2 and P3 contain the respective constraints on $\alpha_{ij}^{(n)}$ and $\beta_{ij}^{(m)}$ respectively, the updated $\alpha_{ij}^{(n)}$ and $\beta_{ij}^{(m)}$ will automatically form a feasible solution for P1.

Let $F \left(\alpha^{(k)}, \beta^{(k)} \right)$ denote the optimization objective in P1. For each optimization step $k$, we have
\begin{equation}
F \left(\alpha^{(k)}, \beta^{(k)} \right) \geq F \left(\alpha^{(k+1)}, \beta^{(k)} \right) \geq F \left(\alpha^{(k+1)}, \beta^{(k+1)} \right),
\end{equation}
and $F \left(\alpha^{(k)}, \beta^{(k)} \right) \geq 0$. By the monotone convergence theorem, we conclude that $F \left(\alpha^{(k)}, \beta^{(k)} \right)$ converge to $F^*$. Furthermore, by the Bolzano-Weierstrass theorem, we conclude that $\alpha^{(k)}$ and $\beta^{(k)}$ have subsequences that converge to $\alpha^*$ and $\beta^*$, respectively. This completes the proof of Proposition~\ref{prop:bcd}.

\bibliographystyle{IEEEtran}
\bibliography{Ref}

\end{document}